\documentclass{article}

\usepackage{arxiv}

\usepackage[utf8]{inputenc} % allow utf-8 input
\usepackage[T1]{fontenc}    % use 8-bit T1 fonts
\usepackage{hyperref}       % hyperlinks
\usepackage{url}            % simple URL typesetting
\usepackage{nicefrac}       % compact symbols for 1/2, etc.
\usepackage{microtype}      % microtypography
\usepackage{lipsum}

\usepackage{xcolor}
\usepackage{cite}
\usepackage{amsmath,amssymb,amsfonts}
\usepackage{booktabs}
\usepackage{graphicx}
\usepackage{textcomp}
\usepackage{bbm}
\usepackage{algorithm}
\usepackage{algorithmicx}
\usepackage{algpseudocode}
\usepackage{textcomp}

\ifodd 0

\newcommand{\rev}[1]{{\color{blue}#1}} %revise of the text
 
\newcommand{\com}[1]{{\color{red}\textbf{Comment}:#1}}%comment of the text
\else

\newcommand{\rev}[1]{#1}

\newcommand{\com}[1]{}
\fi

\newtheorem{definition}{Definition}
\newtheorem{theorem}{Theorem}
\newtheorem{proposition}{Proposition}

\newtheorem{example}{Example}

\newtheorem{assumption}{Assumption}
\newtheorem{proof}{Proof}

\title{Interventions in network games with community structure: local planning, budget allocation and efficiency}

\author{Kun Jin and Mingyan Liu% <-this % stops a space
\thanks{Kun Jin and Mingyan Liu with the Electrical and Computer Engineering Department, University of Michigan, Ann Arbor, MI 48109, USA
        {$\lbrace$\tt\small kunj, mingyan$\rbrace$ @umich.edu}}}% 

\begin{document}
\maketitle

\begin{abstract}

Network games study the strategic interaction of agents connected through a network. Interventions in such a game -- actions a coordinator or planner may take that change the utility of the agents and thus shift the equilibrium action profile -- are introduced to improve the planner's objective. We study the problem of intervention in network games where the network has a group structure with local    planners, each associated with a group. The agents play a non-cooperative game while the planners may or may not have the same optimization objective. We model this problem using a sequential move game where planners make interventions followed by agents playing the intervened game. We provide equilibrium analysis and algorithms that find the subgame perfect equilibrium.  We also propose a two-level efficiency definition to study the efficiency loss of  equilibrium actions in this type of games.

\end{abstract}

\section{INTRODUCTION}

Strategic decision making of (physically or logically) connected agents is often studied as a network game, where the utility of an agent depends on its own actions as well as that of those in its neighborhood as defined by an interaction graph or adjacency matrix. This framework can be used to capture different forms of interdependencies between agents' decisions. Network games and their equilibrium outcomes have been studied in a variety of application areas, including the private provision of public goods \cite{allouch2015private,buckley2006income,khalili2019public}, security decision making in interconnected cyber-physical systems \cite{hota2018interdependent,la2016interdependent}, and shock propagation in financial markets \cite{acemoglu2012network}. 

Within this context, {\em intervention} in a network game typically refers to 
changes in certain game parameters made by a utilitarian welfare maximizer with a budget constraint, who wishes to induce a more socially desirable outcome (in terms of social welfare) under the revised game. A prime example is the study presented in \cite{Galeotti2017}, where intervention takes the form of changing the agents' standalone marginal benefit terms (in a linear quadratic utility model) and changes are costly; this is done by a central/global planner, who wishes to find the set of interventions that lead to the highest equilibrium social welfare subject to a cost constraint. 
%In particular, weighted total effort games. 

Finding optimal interventions could be viewed as a form of mechanism design, because in both cases the design or intervention essentially induces a new game form with desirable properties.  But there are a few distinctions between intervention and the standard mechanism design framework.  Specifically, mechanism design is often not limited to a specific game form, the latter being the outcome of the design, while intervention typically starts from a specified game form and seeks improvement through local changes.  Mechanism design typically has the goal of social optimality (i.e., that the outcome/equilibria of the designed game are social welfare maximizing), while intervention aims to do the best under the constraints of a budget and specified forms of intervention.  

%Technically, the problem of designing intervention in a network game can be treated as a two-stage sequential move game. In the first stage, the planner makes interventions, resulting in an intervened (and perfectly anticipated) game different from the original game; in the second stage, the agents play the intervened game and possibly reach an Nash equilibrium (NE).

In this paper, we are interested in intervention in a network game where the network exhibits a group or community structure and each group or community has its own local group planner.  Since group structures are a common phenomenon across networks of all types, be it social, technological, political, or economic, this modeling consideration allows us to investigate a number of interesting features that often arise in realistic strategic and decentralized decision making.  For instance, a single global budget may be first divided into separate chunks of local budgets at the local planners' disposal; these local budgets may or may not be transferred from one community to another, and the local planners' decisions may or may not take into account the connectivity between themselves and other neighboring communities; local planners may or may not wish to cooperate with each other; and so on.  

%We study different scenarios, in some of them, these different group planners have the same objective, while in other scenarios, each planner may have its unique objective associated with the group and its members.

Of particular interest to our study is the issue of efficiency in this type of decision making systems.  A standard notion used to measure efficiency loss in a strategic game is the Price of Anarchy (POA); this is defined as the upper bound on the ratio of the maximum social welfare (sum utility) divided by the social welfare attained at a Nash equilibrium (NE) of the game.  The numerator is what a social planner aims for, while the denominator is the result of agents optimizing their own utilities and best-responding to each other.  POA has been extensively studied in a variety of games, including in interdependent security games such as \cite{Roughgarden2009, Naghizadeh2014IDS}, where agents' incentive to free-ride or over-consume contributes to the efficiency loss; in routing and congestion games  \cite{Roughgarden2003POA, Christodoulou2005POA}; and in network creation games  \cite{Demaine2012POA}. 

%One particularly relevant notion in this type of games is the efficiency loss. In single-shot games, agents optimize their own utility and generally the Nash equilibrium is different from the action profile that maximizes a social planner's objective the social welfare, which equals the summation of every agent's utility. An upper bound on the ratio of the maximum social welfare divided by the Nash equilibrium social welfare is defined as the Price of anarchy (POA). 

It is not hard to see additional sources of efficiency loss exist in the community intervention problem we are interested in: \rev{in addition to agents' self-interested decision making, local planners' non-cooperation as well as sub-optimal budget allocation among groups can both results in efficiency loss.}
%\rev{that cause the the Level-1 efficiency loss}, local planners may be non-cooperative \rev{and cause the Level-2 efficiency loss}, and budget allocation may \rev{influence the two efficiencies on both levels}.
% be sub-optimal at the community level, interventions by a local planner may be sub-optimal, among others.  Accordingly, in such a multi-scale setting, we shall refer to the efficiency loss caused by agents' non-cooperation the Level-1 efficiency loss, and that cause by local planners' non-cooperation the Level-2 efficiency loss.  We will characterize these losses, as well as the impact of different budget allocation rules (by a single, global planner); we will show that the budget allocation rules can significantly influence the interventions, efficiency and welfare of the game.  
%
The main findings of the paper are summarized as follows:  
\begin{enumerate} 
\item We show that through backward induction the planners can obtain a reduced version of the planners' game that only depends on each other's intervention profiles. Regardless of being cooperative or not, the sequential game always have a unique subgame perfect equilibrium. Moreover, this equilibrium can be achieved through a decentralized algorithm based on the best responses of the planners. 
\item We introduce a two-level definition of efficiency loss that allows us to discuss how the planners' actions influence the outcome of the game separately from the agents' actions, and we show that the efficiency loss due to the planners' non-cooperation can be characterized with the budget constraints and shadow prices. 
% \item We analyze the similarity between the optimal intervention profile with the eigenvectors of the matrices that are closely related to the adjacency matrix of the game that captures important interaction information between agents. 
\item We present numerical results on welfare and  efficiency in several commonly seen types of interaction graphs and commonly used budget allocation rules.
\end{enumerate} 

The remainder of the paper is organized as follows. Section \ref{sec:model} introduces our intervention game model and present the objectives for agents and group planners in different scenarios. Then in section \ref{sec:equilibrium}, we show our analysis and characterization on the subgame perfect equilibrium of the intervention game. In section \ref{sec:efficiency}, we study the Level-1 and Level-2 efficiencies of the subgame perfect equilibrium. We present our numerical experiment results in section \ref{sec:numerical}. Finally, section \ref{sec:conclusion} concludes this work.

% The level-2 efficiency losses can come from a variety of reasons, 
% e.g.,
% \begin{enumerate}
% 	\item Non-transferable budgets;
% 	\item Lack of communication in distributed or decentralized optimization;
% 	\item Inconsistent objectives across different groups.
% \end{enumerate}
% \input{Preliminaries}
\section{Game Model}\label{sec:model}

We consider a network game among $N$ agents, denoted by $a_1, \dots, a_N$, represented by a directed graph $\mathcal{G} = (\mathcal{N}, \mathcal{E})$, where $\mathcal{N} $ is the set of nodes/agents and $\mathcal{E} \subseteq \mathcal{N}\times\mathcal{N}$ the set of edges. 
Let $G = (g_{ij})_{i,j}$ denote the adjacency matrix, assumed to be symmetric and as a convention $g_{ii} = 0$; $g_{ij} \neq 0$ implies dependence between $a_i$ and $a_j$, $i\neq j$. 

Agents are divided into $M$ disjoint communities, the $k$th community denoted as $S_k$ with size $N_k$.
Agent $a_i$ takes an action $x_i \in \mathbf{R}$. Let $\pmb{x}_{-i} = [x_1,\dots, x_{i-1}, x_{i+1}, \dots, x_N]^T$ denote the action profile of all except $a_i$, $\pmb{x}_{S_k} = (x_i)_{a_i \in S_k}$ the action profile of members in community $S_k$, and $\pmb{x}_{-S_k}$ the action profile of all agents other than members of $S_k$.  

We consider a family of games with utility: 
\begin{equation}
    u_i(x_i, \pmb{x}_{-i}, y_i) = (b_i + y_i) x_i - \frac{1}{2} x_i^2 + x_i \bigg( \sum_{j \neq i} g_{ij} x_j \bigg)~, 
\end{equation}
which depends on the action profile %(in network games, only the actions of agents connected to $a_i$ matters, but we use the entire action profile for simplicity), 
and a real valued parameter $y_i$ controlled by a planner. This utility function with intervention is studied in \cite{Galeotti2017, Candogan2011}. The $-\frac{1}{2} x_i^2$ is the individual cost for $a_i$, the $b_i x_i$ term is the initial individual marginal benefit, and $x_i \bigg( \sum_{j \neq i} g_{ij} x_j \bigg)$ models the network influence. The intervention component can be seen as a linear subsidiary (discount) term if $y_i > 0$ and a linear penalty (price) term if $y_i < 0$. In this non-cooperative game, the optimization problem of agent $a_i$ for given intervention is
\begin{equation} \label{eq:agent_prob}
    \underset{x_i}{\text{maximize}} ~~u_i(x_i, \pmb{x}_{-i}, y_i). 
\end{equation}
The NE of the game $\pmb{x}^*$, is the action profile where no agent has an incentive to unilaterally deviate, i.e.,
\begin{equation} \label{eq:NE_defn}
    x_i^* = \underset{x_i}{\text{argmax}} ~~ u_i(x_i, \pmb{x}^*_{-i}).
\end{equation}

We denote the planner for $S_k$ as $p_k$, which has a budget constraint $C_k > 0$: $\sum_{a_i \in S_k} y_i^2 \leq C_k$. We denote $\pmb{y}_{S_k} = (y_i)_{a_i \in S_k}$ as the intervention profile of $p_k$ and $\pmb{y}_{-S_k}$ the intervention profile of planners other than $p_i$.  Denote $Q_k = \{ \pmb{y}_{S_k} | \sum_{a_i \in S_k} y_i^2 \leq C_k\}$; thus $Q_k$ is  nonempty, convex and compact. Finally, $Q = \prod_{i=1}^M Q_i$.

We consider two cases.  In the first, planner $p_k$ is a \textit{group-welfare maximizer}, whose objective is to maximize the sum of its members' utilities at the NE, formally
\begin{equation} \label{eq:planner_obj}
    \underset{\pmb{y}_{S_k} \in Q_k}{\text{maximize}} ~~ U_k(\pmb{x}^*, \pmb{y}_{S_k}) = \sum_{a_i \in S_k} u_i(\pmb{x}^*, y_i),
\end{equation}
and we denote $\pmb{y}_{S_k}^* = \underset{\pmb{y}_{S_k} \in Q_k}{\text{argmax}}~ U_k(\pmb{x}^*, \pmb{y}_{S_k})$. When all planners are group-welfare maximizers, we say they are \textit{non-cooperative}.

In the second case,  planner $p_k$ is a \textit{social-welfare maximizer}, whose objective is to maximize the sum of all agents' utilities at the NE, formally
\begin{equation} \label{eq:planner_obj_soc}
    \underset{\pmb{y}_{S_k} \in Q_k}{\text{maximize}} ~~ U(\pmb{x}^*, \pmb{y}_{S_k}, \pmb{y}_{-S_k}) = \sum_{i = 1}^N u_i(\pmb{x}^*, y_i),
\end{equation}
and we denote $\overline{\pmb{y}}_{S_k} = \underset{\pmb{y}_{S_k} \in Q_k}{\text{argmax}}~ U(\pmb{x}^*, \pmb{y}_{S_k}, \overline{\pmb{y}}_{-S_k})$. When all planners are social-welfare maximizers, we say they are \textit{cooperative}.

% \begin{figure}[t]
% \centering
%   \includegraphics[width=0.65\linewidth]{network.png}
%   \caption{An intervention game with 3 communities.}
%   \label{fig:network}
% \end{figure}

% \begin{figure}[t]
% \centering
%   \includegraphics[width=0.65\linewidth]{system_flow_chart.jpg}
%   \caption{A flow chart of the intervention game.}
%   \label{fig:system_flow_chart}
% \end{figure}

\begin{figure}[t]
\centering
\begin{minipage}{.48\textwidth}
    \centering
	\includegraphics[width=\linewidth]{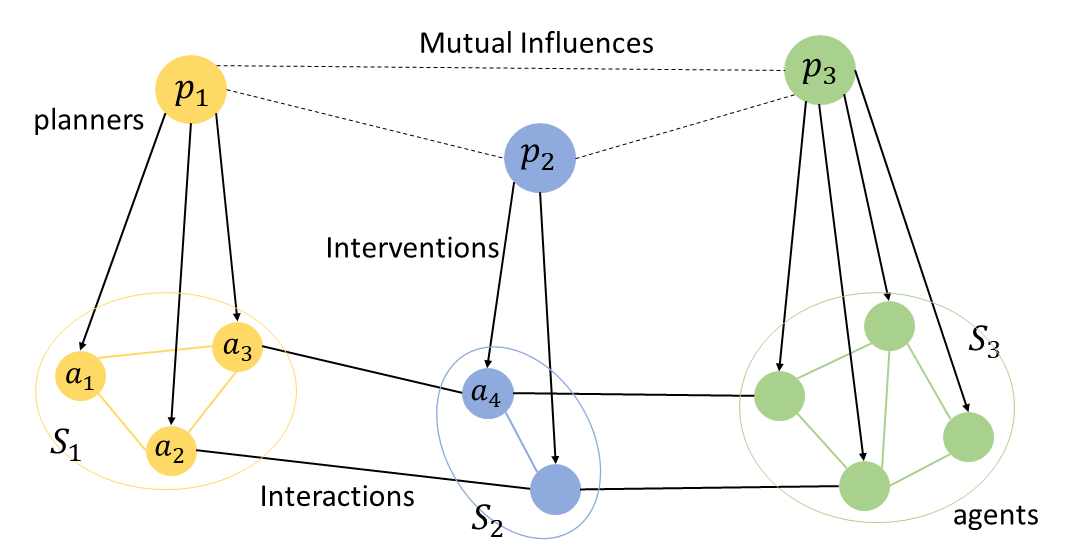}
    \caption{An intervention game with 3 communities.}
    \label{fig:network}
\end{minipage}%
\hspace{2mm}
\begin{minipage}{.48\textwidth}
    \centering
    \includegraphics[width=\linewidth]{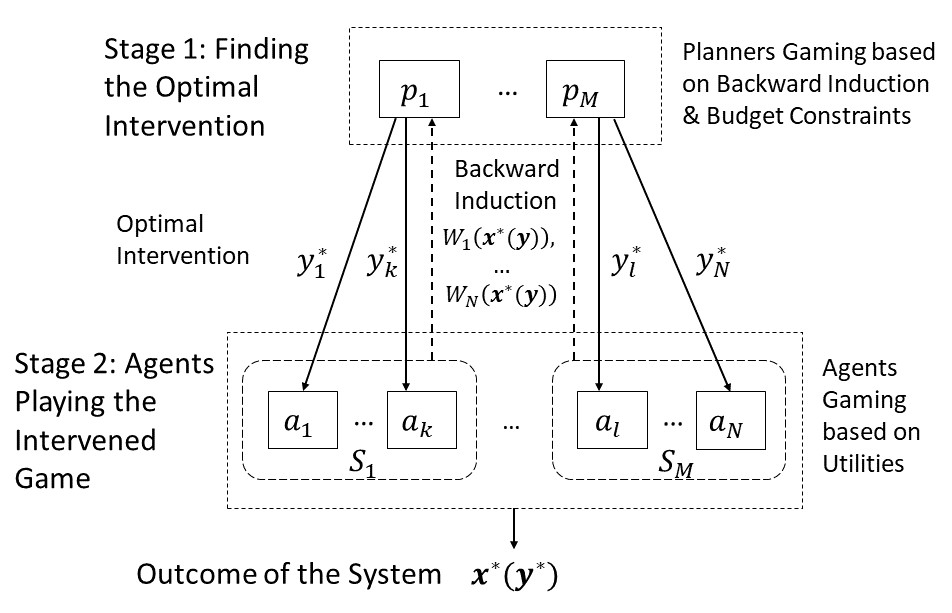}
    \caption{A flow chart of the intervention game.}
    \label{fig:system_flow_chart}
\end{minipage}
\end{figure}

Figure \ref{fig:network} shows the structure of the intervention game described in this section. It's easy to see that with a single planner ($M=1$), the above may be viewed as a two-stage game: the first mover the planner chooses the intervention actions $\mathbf{y}$ in anticipation of the (simultaneous) second movers the agents playing the induced game with actions $\mathbf{x}$. There is a similar two-stage sequentiality in the case of $M$ local planners as shown in Figure \ref{fig:system_flow_chart}: the local planners are simultaneous first movers in choosing interventions for their respective communities, in anticipation of interventions by other local planners and actions by the simultaneous second movers the agents.  For this reason the solution concept we employ this study is the subgame perfect equilibrium. 

%\begin{figure}[t]
%\centering
%\begin{minipage}{.45\linewidth}
%  \centering
%  \includegraphics[width=\linewidth]{noncoop_game.jpg}
%  \caption{Intervention games with non-cooperative planners.}
%  \label{fig:noncoop_game}
%\end{minipage}%
%\hspace{2mm}
%\begin{minipage}{.45\linewidth}
%  \centering
%  \includegraphics[width=\linewidth]{coop_game.jpg}
%  \caption{Intervention games with cooperative planners.}
%  \label{fig:coop_game}
%\end{minipage}
%\end{figure}

%Figure \ref{fig:noncoop_game} and \ref{fig:coop_game} show the flow chart of the intervention games with non-cooperative and cooperative planners respectively.

\section{The subgame perfect equilibrium} \label{sec:equilibrium}
In this section, we characterize the subgame perfect equilibrium of the system and introduce an algorithm to compute it. We assume the following holds throughout this paper:
\begin{assumption} \label{assumption:invertible}
    Matrix $2\text{diag}(\pmb{c}) - W$ is positive definite.
\end{assumption}

We start with computing the NE under an arbitrary intervention, we can compute the first order derivatives as follows
\begin{equation} \label{eq:first_ord_derivatives}
    \frac{\partial u_i}{\partial x_i} = - x_i + \sum_{j \neq i} g_{ij} x_j + (b_i + y_i),
\end{equation}
then by computing the fixed point, we know the unique NE of the game is
\begin{equation} \label{eq:NE}
    \pmb{x}^* = [I - G]^{-1} (\pmb{b} + \pmb{y}), 
\end{equation}
Denote $A = [I - G]^{-2}$ for  simplicity of notation. This NE is known to all planners through backward induction.

\subsection{Finding the subgame perfect equilibrium}

We denote $G_{S_k, S_l}$ as the block of $G$ corresponding to the rows in $S_k$ and columns in $S_l$, and $G_{S_k, \star}$ as the block of $G$ corresponding to the rows in $S_k$ and all columns.
It's worth noting that given the representation of $\pmb{x}^*$ in Eqn (\ref{eq:NE}), the objective of a group welfare maximizer $p_k$ is (See Appendix)
\begin{equation} \label{eq:planner_obj_as_NE}
    U_k(\pmb{x}^*, \pmb{y}_{S_k}) = \frac{1}{2} (\pmb{x}^*_{S_k})^T \pmb{x}^*_{S_k}.
\end{equation}

We can then rewrite the objective of $p_k, \forall k$, and the non-cooperative optimization problem \rev{(P-NC$k$)} as
\begin{align} \label{eq:obj_alt}
    \text{maximize} ~&~ W_k(\pmb{y}) = \frac{1}{2} ||(A^{1/2}(\pmb{y} + \pmb{b}))_{S_k}||_2^2  \nonumber \\
    \text{subject to} ~&~ \sum_{i:a_i \in S_k} (y_i)^2 \leq C_k~. 
\end{align}

It's worth noting that this doesn't imply the planners' optimization problems are independent, since we can write $\pmb{x}^*_{S_k}$ as
$\pmb{x}^*_{S_k} = (A^{1/2}_{S_k,\star} ) (\pmb{y} + \pmb{b})$,
which depends on $\pmb{y}_{-S_k}$ unless $S_k$ is isolated.

Similarly, we can rewrite the cooperative optimization problem \rev{(P-C)} where all planners are social welfare maximizers
\begin{align} \label{eq:soc_opt_prob}
    \text{maximize} ~&~ W(\pmb{y}) = \frac{1}{2} (\pmb{y} + \pmb{b})^T A (\pmb{y} + \pmb{b})  \nonumber \\
    \text{subject to} ~&~ \sum_{i:a_i \in S_k} y_i^2 \leq C_k, ~~\forall k
\end{align}

We can also write out the decentralized version of (P-C) where each planner $p_k$ has its own optimization problem (P-C$k$) given other planners' intervention profile $\pmb{y}_{-S_k}$
\begin{align} \label{eq:soc_opt_prob_pk}
    \text{maximize} ~&~ W(\pmb{y}) = \frac{1}{2} (\pmb{y} + \pmb{b})^T A (\pmb{y} + \pmb{b})  \nonumber \\
    \text{subject to} ~&~ \sum_{i:a_i \in S_k} (y_i)^2 \leq C_k.
\end{align}

We have the following result.  
\begin{theorem} \label{thm:unique_eq}
If all  planners are group-welfare maximizers or if they are all social-welfare maximizers, then in each case there is a unique optimal intervention, i.e., unique subgame perfect equilibrium, and under the optimal intervention, the budget constraints are tight. 
\end{theorem}

This is obvious when all planners are social-welfare maximizers from Eqn (\ref{eq:soc_opt_prob}); for the other case see Appendix.
We also propose the following decentralized algorithm based on best response dynamics (BRD) that computes the subgame perfect equilibrium in both cases. Note that the planners' best-response computation utilizes Eqn (\ref{eq:obj_alt}) and (\ref{eq:soc_opt_prob}).
%\com{Do you want to highlight in the algorithm that a planner also needs to do (7) each time it best responds?} 
%\resp{Actually they can use the QCQP form of their problems and skip the computation of the exact NE in (7).}
\begin{algorithm} 
    \caption{Planners' BRD} 
    \label{alg:compute_eq}
    
    \begin{algorithmic}
        \State Initialize: $\pmb{y}(0) = \pmb{y}_0$, $t = 0$
        \While{$\pmb{y}$ not converged}
            \ForAll {$k = 1:M$}
                \State $\pmb{y}_{S_k}(t+1) = \underset{\pmb{y}_{S_k} \in Q_k}{\text{argmax}}~ U_k( \pmb{y}_{S_k}, \pmb{y}_{-S_k}(t))$ \\
                \Comment{best response w.r.t objective $U_k$, which can be either group or social welfare}%\\
           %\Comment{$U_k$ can either be social welfare or group welfare}
            \EndFor
            \State $t \leftarrow t+1$
        \EndWhile
        \State Set optimal intervention profile as $\pmb{y}^* = \pmb{y}(t)$
        \State Compute the agents' Nash equilibrium in the intervened game $\pmb{x}^* = (I-G)^{-1}(\pmb{b}+\pmb{y}^*)$
    \end{algorithmic}
\end{algorithm}

\begin{theorem} \label{thm:alg_converge}
    If all  planners are group-welfare maximizers or if they are all social-welfare maximizers, then in each case Algorithm \ref{alg:compute_eq} converges to the unique subgame perfect equilibrium under Assumption \ref{assumption:invertible}.
\end{theorem}

\begin{proof}
    \textit{(Sketch)} The proof is based on the Jacobian of the best response mappings of the planners. The Jacobian matrix is positive definite when Assumption \ref{assumption:invertible} is true, and thus there is a unique fixed point and the best response mappings have contraction properties. Therefore, the algorithm  converges to the unique fixed point, i.e., the unique optimal intervention profile, which then leads to the unique subgame perfect equilibrium in the game.
\end{proof}

Please see Appendix for the full proof.

\begin{proposition} \label{prop:noncoop_equivalent}
The following cooperative optimization problem \rev{(P-NC-alt)} has the same optimal intervention outcome as the original non-cooperative problem \rev{(P-NC$k$)} where all planners are group-welfare maximizers: 
\begin{align} \label{eq:non_coop_alt}
    \text{maximize} ~&~ \Tilde{W}(\pmb{y}) = \frac{1}{2} (\pmb{y} + \pmb{b})^T \Tilde{A} (\pmb{y} + \pmb{b})  \nonumber \\
    \text{subject to} ~&~ \sum_{i:a_i \in S_k} y_i^2 \leq C_k, \forall k
\end{align}
where
\begin{equation*}
    \Tilde{A} = \begin{bmatrix} 
    A_{S_1, S_1} & \frac{1}{2} A_{S_1,S_2} & \cdots & \frac{1}{2} A_{S_1, S_M} \\
    \frac{1}{2} A_{S_2, S_1} & A_{S_2,S_2} & \cdots & \frac{1}{2} A_{S_2, S_M} \\
    \vdots & \vdots & \ddots & \vdots \\
    \frac{1}{2} A_{S_M, S_1} & \frac{1}{2} A_{S_M,S_2} & \cdots & A_{S_M, S_M} 
    \end{bmatrix}
\end{equation*}

%The above optimization problem has the same optimal intervention outcome as the system where all planners are group welfare maximizers.
\end{proposition}
\vspace{2mm}

\begin{proof}
    \textit{(Sketch)} This is obtained by studying the gradient of each planner's objective function. The above cooperative optimization problem in Eqn (\ref{eq:non_coop_alt}) has the exact same gradient for every planner in the original non-cooperative optimization problem in Eqn (\ref{eq:obj_alt}).
\end{proof}

Please see the appendix for the full proof. We also characterize the direction of the optimal intervention profile in the appendix.

\subsection{Lagrangian Dual and Shadow Prices}
Next we introduce some concepts related to the Lagrangian dual variables and \textit{shadow prices}, which will be used to characterize efficiency budget sharing in the next section.

Since $Q$ clearly satisfies  Slater's Constraint Qualification, the planners' optimization problems are convex regardless of whether they are group or social welfare maximizers. Then based on the KKT condition, we know that strong duality holds for both cooperative and non-cooperative planners' optimization problems. If we define the Lagrangian as
\begin{equation*}
    L(\pmb{y}, \pmb{\lambda}) = W (\pmb{y}) + \sum_{k=1}^M \lambda_k (C_k - ||\pmb{y}_{S_k}||^2),
\end{equation*}
and  $W(\pmb{y})$ is either social  or group welfare, then  we can obtain an optimal dual   $\lambda_k^*$, the shadow price for $p_k$. We can then equivalently think of $p_k$'s problem as maximizing the above Lagrangian with a cost of  intervention $\lambda_k \sum_{a_i \in S_k} y_i^2$. For convenience of notation, \rev{we use $\lambda_k^*$ (resp. $\overline{\lambda}_k$) to denote the dual optimal variable corresponding to group (resp. social) welfare maximization problem in the (P-NC$k$) (resp. (P-C)) problems for the planners.}  

\section{Efficiency and the budget allocation} \label{sec:efficiency}

In this section, we discuss the efficiency of the subgame perfect equilibria under a fixed budget allocation and then study the impact of different budget allocations on the  equilibrium and its efficiency.

\subsection{Efficiency of the subgame perfect equilibrium}
For conventional single-planner multi-agent systems, the efficiency of an NE is characterized as the ratio of the social objective value in the NE divided by the socially optimal outcome, formally
\begin{equation*}
    e(\pmb{x}^*) = \frac{U(\pmb{x}^*)}{\max_{\pmb{x} \succeq \pmb{0}} U(\pmb{x})},
\end{equation*}
and an upper bound on its reciprocal is referred to as the {\em price of anarchy} (PoA) if the objective $U(\pmb{x}) = \sum_{i=1}^N u_i(\pmb{x})$. This maxima is achievable if the agents' utility functions are strictly individually concave and always have a zero point in the first order derivative.

% We also use $\overline{\mu}_k$ to denote the dual optimal variable corresponding to the optimization problem defined in Eqn (\ref{eq:non_coop_alt}).

The introduction of group planners in our intervention problem means there are now multiple sources of efficiency loss.  Accordingly, we will decompose this into a level-1 (L1) component and a level-2 (L2) component, caused by the non-cooperation of agents and planners, respectively. Following the notation of $\pmb{y}^*$ and $\overline{\pmb{y}}$ in Eqn (\ref{eq:planner_obj}) and (\ref{eq:planner_obj_soc}), we formally define the two efficiency loss measures as
\begin{equation} \label{eq:efficiency_loss}
    e_{L1}(\pmb{y}) = \frac{U(\pmb{x}^*, \pmb{y})}{\max_{\pmb{x}} U(\pmb{x}, \pmb{y})},~~ e_{L2} = \frac{W(\pmb{x}^*, \pmb{y}^*)}{ W(\pmb{x}^*, \overline{\pmb{y}})}.
\end{equation}
Thus the overall efficiency, which resembles the conventional definition, can be written as 
\begin{equation*}
    e(\pmb{x}^*, \pmb{y}^*)  = \frac{U(\pmb{x}^*, \pmb{y}^*)}{\max_{\pmb{x} \succeq \pmb{0}, \pmb{y} \in Q} U(\pmb{x}, \pmb{y})} = e_{L2} \cdot e_{L1}(\overline{\pmb{y}})~.
\end{equation*}
%which relates to the conventional definition of efficiency.

The L1 efficiency has been  well studied in the literature, e.g.,  \cite{2011_Walrand}. For an arbitrary intervention profile $\pmb{y}$, if $I-2G \succ \pmb{0}$, then L1 efficiency can be written as
\begin{equation*}
    e_{L1}(\pmb{y}) = \frac{(\pmb{b}+\pmb{y})^T (I-G)^{-2} (\pmb{b}+\pmb{y})}{(\pmb{b}+\pmb{y})^T (I-2G)^{-2} (\pmb{b}+\pmb{y})}~,  
\end{equation*}
since 
\begin{equation*} 
    \frac{\partial (\sum_{i=1}^N u_i)}{\partial x_i} = - x_i + 2 \sum_{j \neq i} g_{ij} x_j + (b_i + y_i),
\end{equation*}
and by computing the fixed point we know the action profile maximizing the social welfare is $\overline{\pmb{x}} = [I - 2G]^{-1} (\pmb{b} + \pmb{y})$.

We have the following result on the L2 efficiency. 
\begin{theorem} \label{thm:L2_efficiency}
    When $\pmb{b} = \pmb{0}$ or $C_k \gg ||\pmb{b}_{S_k}||_2^2, \forall k$, the welfare in (P-C) can be computed by $W = \sum_{k=1}^M \overline{\lambda}_k C_k$, and a lower bound on the L2 efficiency for a given set of budgets is
    \begin{align} \label{eq:l2_efficiency}
        e_{L2} \geq \frac{\sum_{k=1}^M (2\lambda_k^* - \frac{1}{2} \rho_k)  C_k} {\sum_{k=1}^M \overline{\lambda}_k C_k},
    \end{align}
    where $\rho_k$ denotes the spectral radius of $A_{S_k, S_k}$.
\end{theorem}
\vspace{2mm}

\begin{proof}
    \textit{(Sketch)} When $\pmb{b} = \pmb{0}$ or $C_k \gg ||\pmb{b}_{S_k}||_2^2, \forall k$, the optimal intervention $\pmb{y}^*$ becomes the significant part in deciding the L2 efficiency. The shadow prices and current interventions jointly determines $\pmb{x}^*_{S_k}$. Since the budget is binding, we can replace the lengths of optimal interventions with budget values and thus shadow prices and budgets jointly determine the efficiency.
\end{proof}

Please see the appendix for the full proof.

\subsection{Budget allocation and budget transferability}

\rev{We say the budget is transferable 
%if either a higher level, central coordinator determines the allocations $C_k$, or the planners themselves can decide the allocation of $C_k$. In both cases, 
if the individual budgets $C_k$ are fungible and only the aggregate budget constraint ($C$) has to be satisfied. We say the budget is non-transferable if $C_k$ is fixed and cannot be violated for all $p_k$. When all planners are social welfare maximizers and the budget is transferable, the optimization problem reduces to} 
\begin{align} \label{eq:soc_tran_budget}
    \text{maximize} ~&~ W(\pmb{y}) = \frac{1}{2}(\pmb{y} + \pmb{b})^T A (\pmb{y} + \pmb{b})  \nonumber \\
    \text{subject to} ~&~ \sum_{i = 1}^N y_i^2 \leq C~.
\end{align}
This problem is well studied in \cite{Galeotti2017}.

It is obvious that when group planners are social-welfare maximizing, they have incentives to share the budget since they have a common objective. However, it turns out that even when planners are selfish, group-welfare maximizers, they may still have incentives to share the budget. Intuitively, this is because each group $S_k$ has a decreasing marginal benefit in investing in itself, and if a neighboring group $S_l$ has a strong enough positive externality on $S_k$ and has a relatively low budget $C_l$ compared to $C_k$, then $p_k$ will have an incentive to transfer some of its budget to $p_l$. 

%We propose the following condition on when the planners have incentives for budget sharing.
\begin{proposition} \label{prop:transfer_incentive}
    Between two neighboring groups $S_k$ and $S_l$, where $W_{S_k, S_l} \neq \pmb{0}$, if the following inequality holds, then $p_k$ has an incentive to share its budget with $p_l$: 
    \begin{align*}
        (C_k - C_l) (\triangledown_{\pmb{y}_{S_l}} W_l) ^T \pmb{y}_{S_k} \geq  C_l (\triangledown_{\pmb{y}_{S_k}} W_k) ^T \pmb{y}_{S_k}.
    \end{align*}
\end{proposition}
\vspace{2mm}

Please see the appendix for the proof.

We note that compared to non-transferable budget, transferable budget can enable Pareto superior solutions to the system, where every agent and every planner gets a higher payoff.

\begin{example}
    Consider the following game with only two agents, each as a singleton group, and the following utility functions
    \begin{align*}
        u_1(x_1,x_2,y_1) = x_1 - \frac{1}{2} x_1^2 + \frac{1}{2}x_1 x_2 + y_1 x_1,\\
        u_1(x_1,x_2,y_2) = x_2 - \frac{1}{2} x_2^2 + \frac{1}{2}x_1 x_2 + y_2 x_2.
    \end{align*}
    In this case, we have
    \begin{equation*}
        (I-G)^{-1} = \frac{2}{3} \begin{bmatrix} 2 & 1 \\ 1 & 2 \end{bmatrix}, 
    \end{equation*}
    \begin{equation*}
        x_1^* = 2 + \frac{2}{3}(2 y_1 + y_2), ~~x_2^* = 2 + \frac{2}{3}(y_1 + 2 y_2).
    \end{equation*}
    
    Suppose the initial budget is $C_1 = 25$, $C_2 = 0$. When not sharing the budget, the planners will fully invest in $y_1$ and $y_2$ respectively and the resulting subgame perfect equilibrium is $x_1^* = 2 + \frac{2}{3}(10+0) = \frac{26}{3}, x_2^* = 2 + \frac{2}{3}(5+0) = \frac{16}{3}$. But if $p_1$ shares the budget and make it $C_1 = 16, C_2 = 9$, we will have the equilibrium at $x_1^* = 2 + \frac{2}{3}(8+3) = \frac{28}{3}, x_2^* = 2 + \frac{2}{3}(4+6) = \frac{20}{3}$. So with budget sharing, we obtain a uniformly better outcome for all involved.
\end{example}
\section{Numerical Results} \label{sec:numerical}
We present numerical results in this section.

\subsection{Budget Allocation and Network Types}
Our focus is on examining the L2 efficiency with a number of commonly used budget allocation rules under the following types of networks/interaction graphs. 
\begin{enumerate}
    \item Type 1: strong within-group connection, weak between-group connection.  In this type of networks, groups are used to model local and regional organizations formed by individuals; within each organization, agents interact much more frequently and have higher dependencies on local neighbor's decisions. Mathematically, this means that the diagonal blocks $G_{S_k, S_k}$ have more non-zero elements and the non-zero elements have larger absolute values compared to the off-diagonal blocks $G_{S_k, S_l}$.
    
    \item Type 2: weak within-group connection, strong between-group connection. This is the opposite of Type 1; in this case the off-diagonal blocks are now more frequently filled with larger elements. This type of networks can be used to model logical  connectivity, where a group represents a set of agents playing the same role in a game. For example, in a network of sellers and buyers of a set of goods, a seller may interact more frequently with buyers than another seller. 
    In the extreme case where sellers (resp. buyers) only interact with buyers but not with other sellers (resp. buyers), a multipartite graph can be used to capture their interactions. 
    %, the network can represent supply chains, where each group represent agents of the same role in the chain.
    
    \item Type 3: evenly distributed connections. Here groups become a rather arbitrarily constructed concept that may not correspond to agent interactions in a game.
\end{enumerate}

We will consider the following three types of budget allocation. % and when the budget is not transferable.
\begin{enumerate}
    \item Proportional: each group is assigned a budget proportional to its  size, i.e., $C_k = \frac{N_k}{N} C$
    \item Identical: each group is assigned an equal share of the total budget, i.e., $C_k = C/M$.
    \item Cooperative socially optimal: the allocation in the optimal solution of the cooperative optimization problem ( Eqn (\ref{eq:planner_obj_soc})), where the shadow prices $\lambda_k^*$ are the same for all $k$.
\end{enumerate}

Sample games used in the numerical experiments are generated as follows. In generating a random symmetric $G$, the diagonal elements are set to 0 as previously described in Section \ref{sec:model}. The off-diagonal elements in the diagonal blocks are generated using a Bernoulli distribution with parameter $P^{in}_{exist}$, the probability for an edge (non-zero element) to exist between a pair of agents. The absolute value of a non-zero element (strength of a connection) $|g_{ij}|$ is drawn from a uniform distribution on the interval $[S^{in}_{low}, S^{in}_{high}]$. 
% \com{This is where you spell out the different types of games/connections referred to later in the section..} 
The off-diagonal blocks of $G$ are similarly generated using the same approach, with parameters $P^{out}_{exist}$ and $[S^{out}_{low}, S^{out}_{high}]$, respectively. %and arbitrarily determined signs of strengths. 
\rev{The signs of the connections are assigned to yield the following two types of games. In the first, within-group connections and between-group connections have the same sign (all positive); in the second, they have opposite signs (positive within-group, negative between-group; this is also referred to as conflicting groups below).} 
The $\pmb{b}$ vector is generated by sampling every element uniformly from an interval $[b_{low}, b_{high}]$.

For strong connections, $P_{exist} = 0.8$ and $S_{low} = 0.7, S_{high} = 0.9$. For weak connections, $P_{exist} = 0.2$ and $S_{low} = 0.1, S_{high} = 0.3$. For evenly distributed networks, $P^{in}_{exist} = P^{out}_{exist} = 0.5$ and $S^{in}_{low} = S^{out}_{low} = 0.4, S^{in}_{high} = S^{out}_{high} = 0.6$. We then normalize the generated $G$ by the total number of agents in the game to make sure that Assumption \ref{assumption:invertible} holds \footnote{If the product of the expected connection strengths and the connection frequency is fixed, the results are very similar. For this and brevity reasons we don't show results with combinations of low/high connection frequency with strong/weak connections.}. We also choose $b_{low} = 0.1, b_{high} = 0.5$ to make sure that agents will have an initial incentive to take action above 0 and the budget can easily achieve $C_k \gg ||\pmb{b}_{S_k}||_2^2$. These sample games contain two groups, $S_1$ with 40 agents and $S_2$ with 10 agents; we obtained very similar results  with more groups and thus will focus on this setting for brevity.

% \begin{figure}[t]
% \centering
%   \includegraphics[width=0.5\linewidth]{example.png}
%   \caption{An example of a 15 agent, Type 1 network.}
%   \label{fig:example}
% \end{figure}

\subsection{Social Welfare and L2 Efficiency}
For each network type, we show the social welfare with non-cooperative planners and the L2 efficiency on example games with different types of budget allocation rules.

\begin{figure}[t]
\centering
\begin{minipage}{.23\linewidth}
    \centering
	\includegraphics[width=\textwidth]{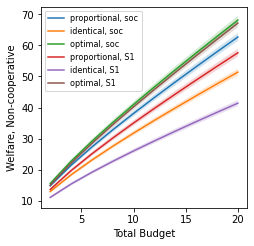}
    \caption{Type 1 all positive network.}
    \label{fig:wel_t1_pos}
\end{minipage}%
\hspace{2mm}
\begin{minipage}{.23\textwidth}
    \centering
	\includegraphics[width=\textwidth]{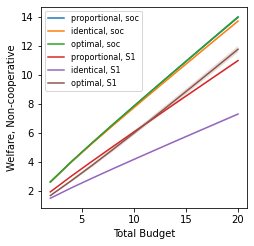}
    \caption{Type 1 conflicting groups.}
    \label{fig:wel_t1_clash}
\end{minipage}
\hspace{2mm}
\begin{minipage}{.23\textwidth}
    \centering
	\includegraphics[width=\textwidth]{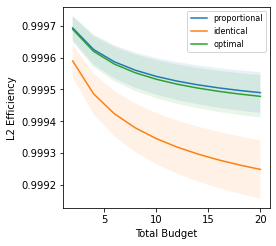}
    \caption{Type 1 all positive network.}
    \label{fig:eff_t1_pos}
\end{minipage}%
\hspace{2mm}
\begin{minipage}{.23\textwidth}
    \centering
	\includegraphics[width=\textwidth]{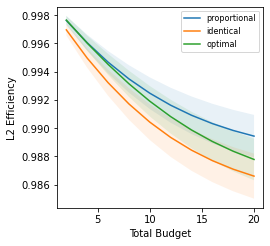}
    \caption{Type 1 conflicting groups.}
    \label{fig:eff_t1_clash}
\end{minipage}
\end{figure}

\begin{figure}[t]
\centering
\begin{minipage}{.23\textwidth}
    \centering
	\includegraphics[width=\textwidth]{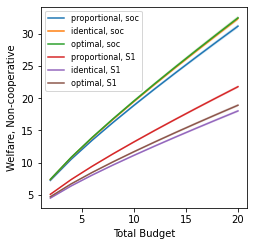}
    \caption{Type 2 all positive network.}
    \label{fig:wel_t2_pos}
\end{minipage}%
\hspace{2mm}
\begin{minipage}{.23\textwidth}
    \centering
	\includegraphics[width=\textwidth]{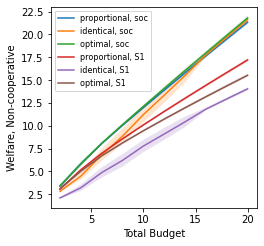}
    \caption{Type 2 conflicting groups.}
    \label{fig:wel_t2_clash}
\end{minipage}
\hspace{2mm}
\begin{minipage}{.23\textwidth}
    \centering
	\includegraphics[width=\textwidth]{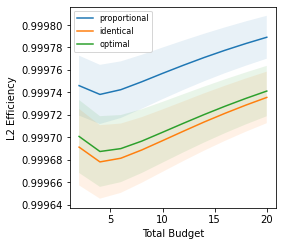}
    \caption{Type 2 all positive network.}
    \label{fig:eff_t2_pos}
\end{minipage}%
\hspace{2mm}
\begin{minipage}{.23\textwidth}
    \centering
	\includegraphics[width=\textwidth]{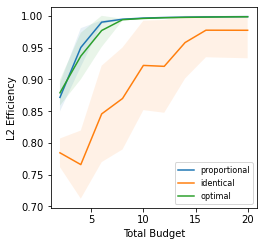}
    \caption{Type 2 conflicting groups.}
    \label{fig:eff_t2_clash}
\end{minipage}
\end{figure}

\begin{figure}[t]
\centering
\begin{minipage}{.45\textwidth}
    \includegraphics[width=\linewidth]{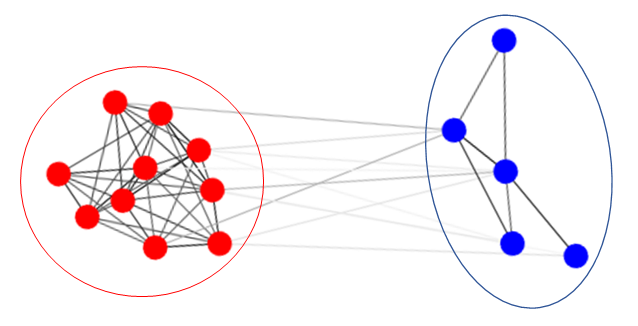}
  \caption{An example of a 15 agent, Type 1 network.}
  \label{fig:example}
\end{minipage}%
\hspace{2mm}
\begin{minipage}{.22\textwidth}
    \centering
	\includegraphics[width=\textwidth]{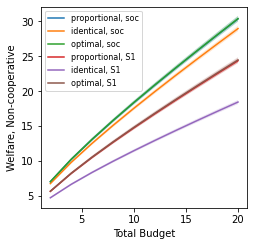}
    \caption{Type 3 network welfare.}
    \label{fig:wel_t3_pos}
\end{minipage}%
\hspace{2mm}
\begin{minipage}{.24\textwidth}
    \centering
	\includegraphics[width=\textwidth]{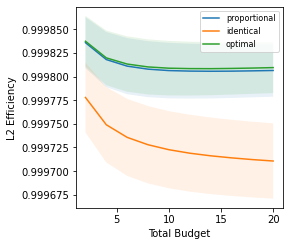}
    \caption{Type 3 network efficiency.}
    \label{fig:eff_t3_pos}
\end{minipage}
\end{figure}

In general, the L2 efficiency is fairly high in all cases except for Type 2 networks with \rev{conflicting groups}. 
Therefore, we only show the welfare results in the (P-C) problem. The main reason for this phenomenon is Assumption \ref{assumption:invertible}, where we require the elements of $G$ to have relatively small values compared to 1 and thus in all except for Type 2 with conflicting groups, the difference between matrices $A$ and $\Tilde{A}$ is small. The major cause of welfare differences come from budget allocation rules.

Figure \ref{fig:wel_t1_pos} and \ref{fig:eff_t1_pos} show the welfare for non-cooperative planners and the L2 efficiency in Type 1 network with all positive connections. In this case, the socially optimal budget allocation yields the highest cooperative and non-cooperative welfare assigns almost all budget to $S_1$.

\rev{Figure \ref{fig:wel_t1_clash} and \ref{fig:eff_t1_clash} show another case of Type 1 network where between-group connections are all negative, but within-group connections remain positive. This can model that the type of interactions between members in the same group are different from agents in different groups. In a special case of this type of network where every agent is taking a positive action level, an increase in an agent's action level can increase (resp. decrease) the agents' utilities in the same group (resp. other groups).} 
% \com{Do you mean between-group connections are negative? This is the part that should be moved earlier in game generation description.  I think you only have two settings: all-positive (within-group and between-group connections) and positive within-group and negative between-group connections, which you call conflicting groups? I think these wordings are fine, we just have to move this to that earlier part and spell them out, and then we don't have to worry about the definition later.} 
In this case, proportional allocation rule is almost socially optimal.

In the Type 2 network, where all connections are positive, Figure \ref{fig:wel_t2_pos} and \ref{fig:eff_t2_pos} show the welfare with non-cooperative planners. Interestingly, the identical budget allocation rule is actually closer to the socially optimal allocation.
For the same Type 2 network but with negative between-group connections, results shown in 
Figure \ref{fig:wel_t2_clash} and \ref{fig:eff_t2_clash} are very different from other network types since the efficiency is now significantly below 1 when we have small budgets. 
% We also find out that now with non-cooperative planners, the proportional allocation can yield a higher social welfare than the socially optimal allocation, and the efficiency loss caused by the non-cooperativeness between planners is substantial.

Figure \ref{fig:wel_t3_pos} and \ref{fig:eff_t3_pos} shows the results in the Type 3 network with all positive connections. In fact, for all combinations of connection signs, the trends are very similar, but the social welfare is significantly lower when we have negative connections. We also see that proportional allocation rule is almost socially optimal.

\rev{Empirically, we observe that for all types of network that when the budget grows larger, under any type of budget allocation rule the welfare  grows approximately linearly and the efficiency approximately converge to a fixed value.} 

We also measured the tightness of the theoretical lower bound on the L2 efficiency. For all above introduced network types except for Type 1 with conflicting groups, the gaps between the lower bounds and the actual L2 efficiencies are less than 0.006, for all three budget allocation rules. For Type 1 network with conflicting groups, the gap is around 0.07 for for all three budget allocation rules. All gaps are  and almost invariant in total budget.
% \com{Perhaps we should add at the beginning: Empirically, we observe that ... instead of "It is true..."?} 

%%%%%%%%%%%%%%%%%%%%%%%%%%%%%%%%%%%%%%%%%%%%%%%%%%%%%%%%%%%%%%%%%%%%%%%%%%%%%%%%

\section{CONCLUSIONS} \label{sec:conclusion}

In this work, we studied an intervention problem in network games with community structures and multiple planners. We showed that given any intervention action, the agents will always have a unique NE. The planners can thus use backward induction and design (locally) optimal interventions. We find that no matter the planners are cooperative or non-cooperative, the system always has a unique subgame perfect equilibrium that fully spends the budget and is Pareto efficient. We also studied the efficiency of the outcomes under different settings in this system, including whether the planners are cooperative and whether the budget is transferable both analytically and numerically. Our analysis shows that we can use the Lagrangian dual optimal variable values to characterize the efficiency, and planners have incentives to share budgets even when they are non-cooperative. The budget transferability also enables uniformly better outcomes than the non-transferable case. Empirically, we observe that the type of network determines which type of (commonly used) budget allocation rule is the most efficient.


%%%%%%%%%%%%%%%%%%%%%%%%%%%%%%%%%%%%%%%%%%%%%%%%%%%%%%%%%%%%%%%%%%%%%%%%%%%%%%%%

%%%%%%%%%%%%%%%%%%%%%%%%%%%%%%%%%%%%%%%%%%%%%%%%%%%%%%%%%%%%%%%%%%%%%%%%%%%%%%%%

\bibliographystyle{unsrt}  
\bibliography{references} 
% \bibliographystyle{plain}
% \bibliography{papers}

%%%%%%%%%%%%%%%%%%%%%%%%%%%%%%%%%%%%%%%%%%%%%%%%%%%%%%%%%%%%%%%%%%%%%%%%%%%%%%%%
\newpage
\section*{APPENDIX}

\subsection{Derivation of Eqn (\ref{eq:planner_obj_as_NE})}
\vspace{-2mm}
\begin{align*}
     U_k(\pmb{x}^*, \pmb{y}_{S_k}) = &~ (\pmb{x}^*_{S_k})^T (-\frac{1}{2}I + G_{S_k, S_k}) \pmb{x}^*_{S_k} + \sum_{l \neq k} (\pmb{x}^*_{S_k})^T G_{S_k, S_l} \pmb{x}^*_{S_l} + (\pmb{b}_{S_k} + \pmb{y}_{S_k})^T \pmb{x}^*_{S_k} \\
    = &~ (\pmb{x}^*_{S_k})^T (-I + G_{S_k, S_k}) \pmb{x}^*_{S_k} + \frac{1}{2} (\pmb{x}^*_{S_k})^T \pmb{x}^*_{S_k} + \sum_{l \neq k} (\pmb{x}^*_{S_k})^T G_{S_k, S_l} \pmb{x}^*_{S_l} + (\pmb{b}_{S_k} + \pmb{y}_{S_k})^T \pmb{x}^*_{S_k}\\ 
    = &~ \frac{1}{2} (\pmb{x}^*_{S_k})^T \pmb{x}^*_{S_k}.
\end{align*}

\subsection{Proof of Theorem \ref{thm:unique_eq} and Theorem \ref{thm:alg_converge}}
Here we introduce some concepts and definitions related to the best responses and the variational inequality(VI) problem.

An agent's best response is the set of strategies that maximizes its utility given the actions taken by all the other agents. Formally, the best response is a set defined by
\begin{equation} \label{eqn:BR}
BR_i(\pmb{x}_{-i}, u_i) = \underset{x_i}{\text{argmax}} ~u_i(x_i, \pmb{x}_{-i}).
\end{equation}

Clearly, an NE of a game is a fixed point of this best response correspondence. One important tool that is useful for analyzing the uniqueness of NE is \emph{Variational Inequalities (VI)}. To establish the connection between NE and VI we assume the utility functions $u_i, \forall i = 1,\dots,N$, for agents $a_1, \dots, a_n$ are continuously twice differentiable. Let $K = \prod_{i=1}^N K_i$ and define $F: \mathbb{R}^N \rightarrow \mathbb{R}^N$ as follows:
\begin{equation} \label{eqn:F(x)}
	F(\pmb{x}) := \bigg( -\triangledown_{x_i} u_i(\pmb{x}) \bigg)_{i=1}^N ~.  
\end{equation}
Then $\pmb{x}^*$ is said to be a solution to VI$(K, F)$ if and only if
\begin{equation} \label{eqn:VI_flat}
	(\pmb{x} - \pmb{x}^*)^T F(\pmb{x}^*) \geq 0, ~  \forall \pmb{x} \in K ~.
\end{equation}
In other words, the solution set to VI$(K,F)$ is equivalent to the set of NE of the game. The following condition can guarantee the uniqueness of NE and the convergence of BRD.
\begin{definition} \label{thm:P_upsilon}
	\textbf{The $P_{\Upsilon}$ condition}: We denote 
	\begin{equation*}
	    \alpha_i(F) = \inf_{\pmb{x} \in K} ||\triangledown_i F_i||_2, ~
	    \beta_{i,j}(F) = \sup_{\pmb{x} \in K} ||\triangledown_j F_i||_2, i \neq j.
	\end{equation*}
	The $\Upsilon$ matrix generated from $F: \mathbb{R}^N \rightarrow \mathbb{R}^N$ is given as follows
	\begin{equation*} %\label{eqn:upsilon}
	\Upsilon(F) = \begin{bmatrix}
	\alpha_1(F) & -\beta_{1,2}(F) & \cdots & -\beta_{1,N}(F) \\
	-\beta_{2,1}(F) & \alpha_2(F) & \cdots & -\beta_{2,N}(F) \\
	\vdots & \vdots & \ddots & \vdots \\
	-\beta_{N,1}(F) & -\beta_{N,2}(F) & \cdots & \alpha_N(F)
	\end{bmatrix}.
\end{equation*}

	If $\Upsilon(F)$ is a P-matrix, that is, if all of its principal components have a positive determinant, then we say $F$ satisfies the $P_{\Upsilon}$ condition. 
	
	In \cite{Scutari2014}, the authors showed that if $F$ satisfies the $P_{\Upsilon}$ condition, $F$ is strongly monotone on $K$, and VI$(K,F)$ has a unique solution. Moreover, the BRD (both synchronous and asynchronous) converges to the unique NE.
\end{definition}

% Based on the above theorem, we provide the proof of Theorem \ref{thm:unique_eq} and Theorem \ref{thm:alg_converge} together
We will show the corresponding $\Upsilon$ matrices in (P-C) and (P-NC-alt) are P-matrices.

\begin{proof}
    We first prove that this is true for the cooperative intervention game. We first fit the planners' game into the VI framework, where $Q$ clearly is the action space, and we can similarly define the operator as
    \begin{equation*}
        F(\pmb{y}) := \bigg( -\triangledown_{y_{S_k}} W(\pmb{y}) \bigg)_{k=1}^M, 
    \end{equation*}
    Then we can define
    \begin{equation*}
        \alpha_k(F) = \inf_{\pmb{y} \in Q} ||\triangledown_k F_k||_2 = ||A_{S_k,S_k}||_2,
    \end{equation*}
    \begin{equation*}
        \beta_{k,l}(F) = \sup_{\pmb{y} \in Q} ||\triangledown_l F_k||_2 = ||A_{S_k,S_l}||_2
    \end{equation*}
    and have the corresponding $\Upsilon$ matrix such that
    \begin{equation*}
         A \succ \pmb{0} \Leftrightarrow \pmb{l}^T \Upsilon(F) \pmb{l} \geq \pmb{y}^T A \pmb{y} > 0, \forall \pmb{y}  
         \Rightarrow
         \Upsilon(F) \succ \pmb{0},
    \end{equation*}
    where $\pmb{l} \in \mathbf{R}^K, l_k = ||\pmb{y}_k||_2$.
    
    $\Upsilon \succ \pmb{0}$ if and only if it's a P-matrix, and thus there is a unique equilibrium in the cooperative planners' intervention game and Algorithm \ref{alg:compute_eq} converges to it. $A \succ \pmb{0}$ and $Q$ is convex and compact also implies the budget tightness.
    
    Similarly,  we can fit the non-cooperative planners' intervention game into the variational inequality framework, where $Q$ is the action space and the operator is
    \begin{equation*}
        F(\pmb{y}) := \bigg( -\triangledown_{y_{S_k}} W_k(\pmb{y}) \bigg)_{k=1}^M, 
    \end{equation*}
    then the $\Upsilon$ matrix is(also shown in proposition \ref{prop:noncoop_equivalent})
    \begin{equation*}
        \Upsilon(F) = \Tilde{A} \succ \pmb{0},
    \end{equation*}
    so there is a unique equilibrium in the cooperative planners' intervention game and Algorithm \ref{alg:compute_eq} converges to it. $\Tilde{A} \succ \pmb{0}$ and $Q$ is convex and compact also implies the budget tightness.
    
    It remains to prove the $\Tilde{A} \succ \pmb{0}$ part above. We note that $A_{S_k,S_k} \succ \pmb{0}, \forall k$ since they are principal minors of $A$ and $A \succ \pmb{0}$. Therefore, for $\forall \pmb{z} \in \mathbf{R}^N, \pmb{z} \neq \pmb{0}$, we have
    \begin{align*}
        2 \pmb{z}^T \Tilde{A} \pmb{z} = \pmb{z}^T A \pmb{z} + \sum_{k=1}^M \pmb{z}_{S_k}^T A_{S_k,S_k} \pmb{z}_{S_k} > 0 \Leftrightarrow \Tilde{A} \succ \pmb{0}.
    \end{align*}
    
\end{proof}

\subsection{Proof of Proposition \ref{prop:noncoop_equivalent}}
\begin{proof}
    We can write out the first order derivative in Eqn (\ref{eq:obj_alt}) for planner $p_k$ in non-cooperative intervention game,
    \begin{align*}
        \triangledown_{\pmb{y}_{S_k}} W_k(\pmb{y}) = A_{S_k, S_k} (\pmb{y}_{S_k} + \pmb{b}_{S_k}) + \frac{1}{2} \sum_{l \neq k} A_{S_k, S_l} (\pmb{y}_{S_l} + \pmb{b}_{S_l}).
    \end{align*}
    
    Similarly, we can write out the first order derivative in Eqn (\ref{eq:soc_opt_prob}) for planner $p_k$ in non-cooperative intervention game,
    \begin{align*}
        \triangledown_{\pmb{y}_{S_k}} W(\pmb{y}) = A_{S_k, S_k} (\pmb{y}_{S_k} + \pmb{b}_{S_k}) + \sum_{l \neq k} A_{S_k, S_l} (\pmb{y}_{S_l} + \pmb{b}_{S_l}).
    \end{align*}
    
    So we can see that if the planners cooperatively play an intervention game such that the $A$ matrix is replaced by the $\Tilde{A}$ matrix, we have
    \begin{align*}
        \triangledown_{\pmb{y}_{S_k}} \Tilde{W}(\pmb{y}) = &~ \Tilde{A}_{S_k, S_k} (\pmb{y}_{S_k} + \pmb{b}_{S_k})+ \sum_{l \neq k} \Tilde{A}_{S_k, S_l} (\pmb{y}_{S_l} + \pmb{b}_{S_l})\\
        = &~ A_{S_k, S_k} (\pmb{y}_{S_k} + \pmb{b}_{S_k}) + \frac{1}{2} \sum_{l \neq k} A_{S_k, S_l} (\pmb{y}_{S_l} + \pmb{b}_{S_l})\\
        = &~ \triangledown_{\pmb{y}_{S_k}} W_k(\pmb{y}).
    \end{align*}
    
    Since in both Eqn (\ref{eq:obj_alt}) and (\ref{eq:non_coop_alt}), every planner has the same gradient $\forall \pmb{y} \in Q$, and the intervention action space are the same, we can conclude that solving the original non-cooperative intervention game and the alternative cooperative intervention game with $\Tilde{A}$ are equivalent. In other words, for any arbitrary starting point, the trajectories of BRD will always be the same for the planners, and will converge to the same unique equilibrium in these two problems and thus they are equivalent.
\end{proof}

\subsection{Direction of the optimal intervention}

We first introduce the cosine similarity that defines the similarity of of two vectors in their directions, formally,
\begin{equation}
    \rho(\pmb{r}, \pmb{s}) = \frac{\pmb{r}^T \pmb{s}}{||\pmb{r}||_2 ||\pmb{s}||_2},
\end{equation}
when $\rho(\pmb{r}, \pmb{s}) = 1$, the two vectors have the same direction and $\pmb{r} = \alpha \pmb{s}$ for some $\alpha > 0$.

For non-cooperative planners, we have the following results on the optimal intervention.

\begin{proposition} \label{prop:coop_planners_eigvec}
    For an arbitrary fixed budget constraint allocation, when $\pmb{b} = \pmb{0}$, $\rho(\pmb{y}^*, \pmb{v}_{max}) = 1$, otherwise if $C_k \gg ||\pmb{b}_{S_k}||_2^2, \forall k$, $\rho(\pmb{y}^*, \pmb{v}_{max}) \approx 1$, where $\pmb{y}^*$ is the non-cooperative optimal intervention to the problem in Eqn (\ref{eq:planner_obj}), and $\pmb{v}_{max}$ is the eigenvector that corresponds to the largest eigenvalue of $\Tilde{B}$ matrix, which is defined as
    \begin{equation*}
        \Tilde{B} = \begin{bmatrix} 
        B_{S_1, S_1} & \frac{1}{2} B_{S_1,S_2} & \cdots & \frac{1}{2} B_{S_1, S_M} \\
        \frac{1}{2} B_{S_2, S_1} & B_{S_2,S_2} & \cdots & \frac{1}{2} B_{S_2, S_M} \\
        \vdots & \vdots & \ddots & \vdots \\
        \frac{1}{2} B_{S_M, S_1} & \frac{1}{2} B_{S_M,S_2} & \cdots & B_{S_M, S_M} 
        \end{bmatrix},
    \end{equation*}
    where $B = A^{1/2} D A^{1/2}, ~~D = \text{diag}((\frac{1}{\lambda^*_k} \pmb{1}_{N_k})_{k=1}^M)$.
\end{proposition}

% Proof of Proposition \ref{prop:coop_planners_eigvec}
\begin{proof}
    We begin with the case where $\pmb{b} = \pmb{0}$, and consider the following non-cooperative objective for $p_k$,
    \begin{align} 
        \text{maximize} ~&~ \hat{W}_k(\pmb{y}) = \frac{1}{2}||(B^{1/2} \pmb{y})_{S_k}||_2^2 = \frac{1}{2\lambda_k^*} ||\pmb{x}^*_{S_k}||_2^2  \nonumber \\
        \text{subject to} ~&~ \sum_{i:a_i \in S_k} (y_i)^2 \leq C_k,
    \end{align}
    clearly, we have
    \begin{equation*}
        \triangledown_{\pmb{y}_{S_k}} \hat{W}_k \big|_{\pmb{y} = \pmb{y}^*} = \frac{1}{\lambda_k^*} \triangledown_{\pmb{y}_{S_k}} W_k \big|_{\pmb{y} = \pmb{y}^*}.
    \end{equation*}
    
    From the definition of $\pmb{\lambda}^*$ and $\pmb{y}^*$, we have
    \begin{equation*}
        \triangledown_{\pmb{y}_{S_k}} W_k \big|_{\pmb{y} = \pmb{y}^*} - 2 \lambda_k^* \pmb{y}_{S_k}^* = \pmb{0},
    \end{equation*}
    and thus
    \begin{equation*}
        \triangledown_{\pmb{y}_{S_k}} \hat{W}_k \big|_{\pmb{y} = \pmb{y}^*} - 2 \pmb{y}_{S_k}^* = \pmb{0}.
    \end{equation*}
    
    We can then define the following problem
    \begin{align*} 
        \text{maximize} ~&~ \Tilde{W}_B(\pmb{y}) = \frac{1}{2} \pmb{y}^T \Tilde{B} \pmb{y}\\
        \text{subject to} ~&~ \sum_{i = 1}^N (y_i)^2 \leq C,
    \end{align*}
    similar to proposition \ref{prop:noncoop_equivalent}, we know that 
    \begin{equation*}
        \triangledown_{\pmb{y}_{S_k}} \hat{W}_k = \triangledown_{\pmb{y}_{S_k}} \Tilde{W}_B(\pmb{y}),
    \end{equation*}
    and we know from the KKT conditions that since $||\pmb{y}^*||_2^2 = C$, we have $\pmb{y} = \pmb{y}^*, \mu^* = 1$ as the primal and dual optimal values to the Lagrangian
    \begin{equation*}
        L_B(\pmb{y}, \mu) = \Tilde{W}_B(\pmb{y}) - \mu(C - ||\pmb{y}^*||_2^2),
    \end{equation*}
    and thus since $B = A^{1/2} D A^{1/2}, A \succ \pmb{0}, D \succ \pmb{0}$, we know $B \succ \pmb{0}$ and then similar to the proof of proposition \ref{prop:noncoop_equivalent}, $\Tilde{B} \succ \pmb{0}$. Since $\pmb{y}^*$ is the primal optimal, we know that $\rho(\pmb{y}^*, \pmb{v}_{max}) = 1$, where $\pmb{v}_{max}$ is the eigenvector of $\Tilde{B}$'s largest eigenvalue.
    
    When $\pmb{b} \neq \pmb{0}$ but $C_k >> ||\pmb{b}_{S_k}||_2^2$, it's not hard to show that
    \begin{equation*}
        (\pmb{y}+\pmb{b})^T A (\pmb{y}+\pmb{b}) \rightarrow \pmb{y}^T A \pmb{y}, \text{ as } ||\pmb{y}_{S_k}||/||\pmb{b}_{S_k}|| \rightarrow \infty, \forall k,
    \end{equation*}
    an thus the above results still hold.
\end{proof}

In \cite{Galeotti2017}, the single planner's problem can be thought of multiple cooperative planners with transferable budget, then the shadow prices for every planner is the same in the optimal intervention, and thus $\rho(\pmb{y}^*, \pmb{v}_{max}) = 1$, where $\pmb{v}_{max}$ is the eigenvector of $A$'s largest eigenvalue. Moreover, in that case, when $G$ is an all positive(negative) matrix, $\pmb{v}_{max}$ is the eigenvector of $G$'s largest eigenvalue($-G$'s smallest eigenvalue).

\subsection{Proof of Theorem \ref{thm:L2_efficiency}}
\begin{proof}
    We will provide the derivations for the numerator and denominator separately. And we begin with the case where $\pmb{b} = \pmb{0}$.
    
    The denominator corresponds to the cooperative planners' intervention game, we can write out the Lagrangian function as follows
    \begin{equation*}
        L(\pmb{y}, \pmb{\lambda}) = W(\pmb{y}) + \sum_{k=1}^M \lambda_k(C_k - ||\pmb{y}_{S_k}||^2).
    \end{equation*}
    As mentioned earlier in section \ref{sec:efficiency}, the strong duality holds, and from the KKT conditions, we have at the optimum that
    \begin{align*}
        \triangledown_{\pmb{y}_{S_k}} W \big|_{\pmb{y} = \overline{\pmb{y}}} - 2 \overline{\lambda}_k \overline{\pmb{y}}_{S_k} = \pmb{0}
        ~\Leftrightarrow~ A_{S_k,S_k} \overline{\pmb{y}}_{S_k} + \sum_{l \neq k} A_{S_k,S_l} \overline{\pmb{y}}_{S_l} = 2 \overline{\lambda}_k \overline{\pmb{y}}_{S_k}.
    \end{align*}
    So we can rewrite the social welfare as
    \begin{align*}
        W(\overline{\pmb{y}}) = &~ \frac{1}{2} \overline{\pmb{y}}^T A \overline{\pmb{y}} \\
        = &~ \frac{1}{2} \sum_{k=1}^M \overline{\pmb{y}}_{S_k}^T \bigg(\sum_{l=1}^M A_{S_k,S_l} \overline{\pmb{y}}_{S_l} \bigg)\\
        = &~ \frac{1}{2} \sum_{k=1}^M \overline{\pmb{y}}_{S_k}^T (2 \overline{\lambda}_k \overline{\pmb{y}}_{S_k})\\
        = &~ \sum_{k=1}^M \overline{\lambda}_k ||\overline{\pmb{y}}||_2^2\\
        = &~ \sum_{k=1}^M \overline{\lambda}_k C_k.
    \end{align*}
    
    For the numerator that corresponds to the non-cooperative planners' intervention game, we can write out the Lagrangian function for $p_k$ as follows
    \begin{equation*}
        L_k(\pmb{y}_{S_k}, \pmb{y}_{-S_k}, \lambda_k) = W_k(\pmb{y}_{S_k}, \pmb{y}_{-S_k}) + \lambda_k(C_k - ||\pmb{y}_{S_k}||^2).
    \end{equation*}
    From the KKT conditions, we have at the equilibrium that
    \begin{align*}
        &~ \triangledown_{\pmb{y}_{S_k}} W_k \big|_{\pmb{y} = \pmb{y}^*} - 2 \lambda^*_k \pmb{y}^*_{S_k} = \pmb{0}\\
        \Leftrightarrow &~ A_{S_k,S_k} \pmb{y}^*_{S_k} + \frac{1}{2} \sum_{l \neq k} A_{S_k,S_l} \pmb{y}^*_{S_l} = 2 \lambda^*_k \pmb{y}^*_{S_k}.
    \end{align*}
    So we can rewrite the social welfare as
    \begin{align*}
        W(\pmb{y}^*) = &~ \frac{1}{2} (\pmb{y}^*)^T A \pmb{y}^* \\
        = &~ \frac{1}{2} \sum_{k=1}^M (\pmb{y}^*_{S_k})^T \bigg(\sum_{l=1}^M A_{S_k,S_l} \pmb{y}^*_{S_l} \bigg)\\
        = &~ \frac{1}{2} \sum_{k=1}^M (\pmb{y}^*_{S_k})^T (4\lambda^*_k \pmb{y}^*_{S_k} - A_{S_k,S_k} \pmb{y}^*_{S_k})\\
        \geq &~ \sum_{k=1}^M (2\lambda^*_k - \frac{1}{2} \rho_k) ||\pmb{y}^*_{S_k}||_2^2\\
        = &~ \sum_{k=1}^M (2\lambda^*_k - \frac{1}{2} \rho_k) C_k,
    \end{align*}
    and since $e_{L2} = \frac{W(\pmb{y}^*)}{W(\overline{\pmb{y}})}$, we know the lower bound holds. 
    
    When $\pmb{b} \neq \pmb{0}$ but $C_k >> ||\pmb{b}_{S_k}||_2^2$, it's not hard to show that
    \begin{equation*}
        (\pmb{y}+\pmb{b})^T A (\pmb{y}+\pmb{b}) \rightarrow \pmb{y}^T A \pmb{y}, \text{ as } ||\pmb{y}_{S_k}||/||\pmb{b}_{S_k}|| \rightarrow \infty, \forall k,
    \end{equation*}
    an thus the above results still hold.
\end{proof}

\subsection{Proof of Proposition \ref{prop:transfer_incentive}}
\begin{proof}
    Suppose now group $k$ transfers an infinitesimal amount of budget to group $l$, then after the transfer, the new intervention profile of group $k$ and $l$ becomes
    \begin{equation*}
        \pmb{y}_{S_k}' = (1 - \delta_k) \pmb{y}_k, ~~ \pmb{y}_{S_l}' = (1 + \delta_l) \pmb{y}_l, ~~ \frac{\delta_l}{\delta_k} = \frac{||\pmb{y}_k||^2}{||\pmb{y}_l||^2} = \frac{C_k}{C_l}.
    \end{equation*}
    Then we look at group $k$'s welfare after the transfer, if it increases, group $k$ has an incentive to do the transfer
    \begin{align*}
        W_k(\pmb{y}_{S_k}', \pmb{y}_{S_l}', \pmb{y}_{-(S_k,S_l)}) - W_k(\pmb{y})
        = &~ \frac{1}{2} \big( [(1 + \delta_l)(1 - \delta_k)-1] \pmb{y}_{S_k}^T A_{S_k, S_l} \pmb{y}_{S_l} \\
        &~ + [(1 - \delta_k)^2 - 1] \pmb{y}_{S_k}^T A_{S_k, S_k} \pmb{y}_{S_k}\\
        &~ + 2[(1 - \delta_k) -1] \pmb{y}_{S_k}^T (A_{S_k, :}) \pmb{b} \big) \geq 0\\
        \Leftrightarrow &~ (\delta_l - \delta_k) (\triangledown_{\pmb{y}_{S_l}} W_k)^T \pmb{y}_{S_l} \geq \delta_k (\triangledown_{\pmb{y}_{S_k}} W_k)^T \pmb{y}_{S_k}\\
        \Leftrightarrow &~ (C_k - C_l) (\triangledown_{\pmb{y}_{S_l}} W_k)^T \pmb{y}_{S_l} \geq C_l (\triangledown_{\pmb{y}_{S_k}} W_k)^T \pmb{y}_{S_k}.
    \end{align*}
    since $\delta_l$ and $\delta_k$ are all infinitesimal, we can ignore their second order products in the above derivations.
\end{proof}

% \section*{ACKNOWLEDGMENT}
% This work is supported by the NSF under grants CNS-1939006, CNS-2012001, IIS-1905558 (CAREER) and by the ARO under contract W911NF1810208.

% \bibliographystyle{unsrt}  
% \bibliography{references}  %%% Remove comment to use the external .bib file (using bibtex).


% George Kour and Raid Saabne.
% \newblock Real-time segmentation of on-line handwritten arabic script.
% \newblock In {\em Frontiers in Handwriting Recognition (ICFHR), 2014 14th
%   International Conference on}, pages 417--422. IEEE, 2014.

% George Kour and Raid Saabne.
% \newblock Fast classification of handwritten on-line arabic characters.
% \newblock In {\em Soft Computing and Pattern Recognition (SoCPaR), 2014 6th
%   International Conference of}, pages 312--318. IEEE, 2014.

% \bibitem{hadash2018estimate}
% Guy Hadash, Einat Kermany, Boaz Carmeli, Ofer Lavi, George Kour, and Alon
% \newblock Estimate and replace: A novel approach to integrating deep neural
%   networks with existing applications.
% \newblock {\em arXiv preprint arXiv:1804.09028}, 2018.

% \end{thebibliography}

\end{document}